\documentclass[12pt]{article}
\usepackage{amsxtra}
\usepackage{amssymb, amsmath}
\usepackage{amsthm}

\newtheorem{lemma}{Lemma}[section]
\newtheorem{theorem}{Theorem}[section]
\newtheorem{remark}{Remark}[section]

\newtheorem{proposition}{Proposition}[section]

\numberwithin{equation}{section}

\usepackage[dvips]{graphicx}

\begin{document}
\begin{center}
\textbf{\Large{Inequalities Among Logarithmic-Mean Measures}}
\end{center}

\bigskip
\begin{center}
\textbf{\large{Inder Jeet Taneja}}\\
Departamento de Matem\'{a}tica\\
Universidade Federal de Santa Catarina\\
88.040-900 Florian\'{o}polis, SC, Brazil.\\
\textit{e-mail: ijtaneja@gmail.com\\
http://www.mtm.ufsc.br/$\sim$taneja}
\end{center}

\begin{abstract}
\textit{In this paper we shall consider some famous means such as
arithmetic, harmonic, geometric, logarithmic means, etc. Inequalities involving logarithmic mean with differences among other means are presented.}
\end{abstract}

\bigskip
\textbf{Key words:} {Arithmetic mean, geometric mean, harmonic mean, logarithmic mean and inequalities.}

\bigskip
\textbf{AMS Classification:} 94A17; 62B10.

\section{Introduction}

Let us consider the following well known \textit{mean of order t} \cite{beb}:
\begin{equation}
\label{eq1}
B_t (a,b) = \left\{ {{\begin{array}{*{20}c}
 {\left( {\frac{a^t + b^t}{2}} \right)^{1 / t},} & {t \ne 0} \\
 {\sqrt {ab} ,} & {t = 0} \\
 {\max \{a,b\},} & {t = \infty } \\
 {\min \{a,b\},} & {t = - \infty } \\
\end{array} }} \right.
\end{equation}

\noindent
for all $a,b,t \in {\rm R},\mbox{ }a,b > 0$.

\bigskip
In particular, we have
\begin{align}
B_{ - 1} (a,b) & = H(a,b) = \frac{2ab}{a + b},\notag\\
B_0 (a,b) & = G(a,b) = \sqrt {ab} ,\notag\\
B_{1 / 2} (a,b) & = N_1 (a,b) = \left( {\frac{\sqrt a + \sqrt b }{2}} \right)^2,\notag\\
B_1 (a,b) & = A(a,b) = \frac{a + b}{2},\notag\\
 \intertext{and}  B_2 (a,b) &  = S(a,b) = \sqrt {\frac{a^2 + b^2}{2}} .\notag
\end{align}

\bigskip
The means, $H(a,b)$, $G(a,b)$, $A(a,b)$ and $S(a,b)$ are known in the
literature as \textit{harmonic}, \textit{geometric}, \textit{arithmetic} and \textit{root-square means} respectively. For simplicity we can call the measure, $N_1 (a,b)$ as \textit{square-root} mean. It is well know that \cite{beb} the \textit{mean of order }$t$ given in  (\ref{eq1}) is monotonically increasing in $t$, then we can write

\begin{equation}
\label{eq2}
H(a,b) \le G(a,b) \le N_1 (a,b) \le A(a,b) \le S(a,b).
\end{equation}

\bigskip
Dragomir and Pearce \cite{drp} proved the following inequality:

\[
\frac{a^r + b^r}{2} \le \frac{b^{r + 1} - a^{r + 1}}{(r + 1)(b - a)} \le
\left( {\frac{a + b}{2}} \right)^r,
\]

\noindent
for all $a,b > 0$, $a \ne b$, $r \in (0,1)$. In particular take $r =
\frac{1}{2}$ in (\ref{eq1}), we get
\begin{equation}
\label{eq3}
\frac{\sqrt a + \sqrt b }{2} \le \frac{2(b^{3 / 2} - a^{3 / 2})}{3(b - a)}
\le \sqrt {\frac{a + b}{2}} ,
\quad
a \ne b.
\end{equation}

After necessary calculations in (\ref{eq3}), we get
\begin{equation}
\label{eq4}
\left( {\frac{\sqrt a + \sqrt b }{2}} \right)^2 \le \frac{a + \sqrt {ab} +
b}{3} \le \left( {\frac{\sqrt a + \sqrt b }{2}} \right)\left( {\sqrt
{\frac{a + b}{2}} } \right).
\end{equation}

On the other side we can easily check that
\begin{equation}
\label{eq5}
\left( {\frac{\sqrt a + \sqrt b }{2}} \right)\left( {\sqrt {\frac{a + b}{2}}
} \right) \le \frac{a + b}{2}.
\end{equation}

The expressions  (\ref{eq3}), (\ref{eq4}) and (\ref{eq5}) lead us to the following
inequality:
\begin{equation}
\label{eq6}
H(a,b) \le G(a,b) \le N_1 (a,b)
 \le N_3 (a,b) \le N_2 (a,b) \le A(a,b) \le S(a,b),
\end{equation}

\noindent
where
\begin{equation}
\label{eq7}
N_2 (a,b) = \left( {\frac{\sqrt a + \sqrt b }{2}} \right)\left( {\sqrt
{\frac{a + b}{2}} } \right),
\end{equation}

\noindent
and
\begin{equation}
\label{eq8}
N_3 (a,b) = \frac{a + \sqrt {ab} + b}{3}.
\end{equation}

Thus we have three new means, where $N_1 (a,b)$ appears as a natural way.
The $N_2 (a,b)$ can be seen in Taneja \cite{tan2}, \cite{tan3} and the mean $N_3 (a,b)$ can
be seen in Zhang and Wu  \cite{zhw}. Some more properties of the means given
in (\ref{eq6}) are studied by  \cite{szd},  \cite{sjz}. A survey on these inequalities can be seen in
\cite{tan5}.

\bigskip
Based on the inequalities (\ref{eq6}), the author \cite{tan3}-\cite{tan4}  proved the following
results
\begin{align}
\label{eq9}
H(a,b)& \le \frac{2A(a,b)H(a,b)}{A(a,b) + H(a,b)} \le G(a,b) \le
\frac{2H(a,b) + S(a,b)}{3}\notag\\
& \le \frac{A(a,b) + H(a,b)}{2} \le \sqrt {\frac{\left( {A(a,b)} \right)^2 +
\left( {H(a,b)} \right)^2}{2}} \le \frac{S(a,b) + G(a,b)}{2}\notag\\
& \le \frac{H(a,b) + 2S(a,b)}{3} \le A(a,b) \le S(a,b) + H(a,b) - G(a,b)\notag\\
& \le S(a,b) \le 3\left[ {A(a,b) - G(a,b)} \right] + H(a,b)
\end{align}
\begin{align}
\label{eq10}
H(a,b) &\le G(a,b) \le \frac{G(a,b) + 2N_2 (a,b)}{3} \le N_1 (a,b) \le
\frac{2A(a,b) + 7N_1 (a,b)}{9}\notag\\
& \le N_2 (a,b) \le \frac{A(a,b) + N_1 (a,b)}{2} \le \frac{7A(a,b) +
H(a,b)}{8} \le A(a,b)
\end{align}
\begin{align}
\label{eq11}
G(a,b) &\le \frac{S(a,b) + 3G(a,b)}{4} \le N_1 (a,b)  \le {\frac{S(a,b) + 8N_1 (a,b)}{9}}  \notag\\
 &\le N_3 (a,b) \le N_2 (a,b)  \le \frac{A(a,b) + N_1 (a,b)}{2} \notag\\
 &\le \frac{S(a,b) + 2N_1 (a,b)}{3}  \le \frac{S(a,b) + 4N_2 (a,b)}{5} \le A(a,b).
\end{align}

The above inequalities are based on the following two lemmas \cite{tan2}:
\begin{lemma}
Let $f:I \subset {\rm R}_ + \to {\rm R}$ be a convex and
differentiable function satisfying $f(1) = f^\prime (1) = 0$. Consider a
function
\begin{equation}
\label{eq12}
\phi _f (a,b) = af\left( {\frac{b}{a}} \right),
\quad b > a > 0,
\end{equation}

\noindent
then the function $\phi _f (a,b)$ is convex in ${\rm R}_ + ^2 $, and
satisfies the following inequality:
\begin{equation}
\label{eq13}
0 \le \phi _f (a,b) \le \left( {\frac{b - a}{a}} \right)\phi _{f}' (a,b).
\end{equation}
\end{lemma}

\begin{lemma}
Let $f_1 ,f_2 :I \subset {\rm R}_ + \to {\rm R}$ be two
convex functions satisfying the assumptions:

(i) $f_1 (1) = f_1 ^\prime (1) = 0$, $f_2 (1) = f_2 ^\prime (1) = 0$;

(ii) $f_1 $ and $f_2 $ are twice differentiable in ${\rm R}_ + $;

(iii) there exists the real constants $\alpha ,\beta $ such that $0 \le
\alpha < \beta $ and
\begin{equation}
\label{eq14}
\alpha \le \frac{f_1 ^{\prime \prime }(x)}{f_2 ^{\prime \prime }(x)} \le
\beta ,
\quad
f_2 ^{\prime \prime }(x) > 0,
\end{equation}

\noindent
for all $x > 0$ then we have the inequalities:
\begin{equation}
\label{eq15}
\alpha \mbox{ }\phi _{f_2 } (a,b) \le \phi _{f_1 } (a,b) \le \beta \mbox{
}\phi _{f_2 } (a,b),
\end{equation}

\noindent
for all $a,b \in (0,\infty )$.
\end{lemma}

The following inequality involving \textit{log-mean} is also known in the literature:
\begin{equation}
\label{eq16}
H(a,b) \le G(a,b) \le L(a,b) \le N_1 (a,b),
\end{equation}

\noindent
where, $L(a,b) = \frac{b - a}{\ln b - \ln a}$, $b \ne a$ is the well-known
\textit{logarithmic mean}.

\bigskip
Finally, the expressions (\ref{eq6}) and (\ref{eq16}) lead us to the following
inequality:
\begin{equation}
\label{eq17}
H(a,b) \le G(a,b) \le L(a,b) \le N_1 (a,b)
 \le N_3 (a,b) \le N_2 (a,b) \le A(a,b) \le S(a,b).
\end{equation}

\section{Difference of Means}

Let us rewrite the inequalities (\ref{eq13}) as
\begin{equation}
\label{eq18}
L(a,b) \le N_1 (a,b)
 \le N_3 (a,b) \le N_2 (a,b) \le A(a,b) \le S(a,b).
\end{equation}

In view of (\ref{eq18}), let us consider the following nonnegative differences
\begin{equation}
\label{eq19}
M_{SL} (a,b) = S(a,b) - L(a,b),
\end{equation}
\begin{equation}
\label{eq20}
M_{AL} (a,b) = A(a,b) - L(a,b),
\end{equation}
\begin{equation}
\label{eq21}
M_{N_2 L} (a,b) = N_2 (a,b) - L(a,b),
\end{equation}
\begin{equation}
\label{eq22}
M_{N_3 L} (a,b) = N_3 (a,b) - L(a,b),
\end{equation}

\noindent and
\begin{equation}
\label{eq23}
M_{N_1 L} (a,b) = N_1 (a,b) - L(a,b).
\end{equation}

The inequalities (\ref{eq18}) still admits more nonnegative differences, but they
are already studied before.

\bigskip
Now, we shall prove the convexity of the above differences (\ref{eq19})-(\ref{eq23}).

\begin{theorem}
The means given by (\ref{eq19})-(\ref{eq23}) are nonnegative and
convex in ${\rm R}_ + ^2 $.
\end{theorem}

We shall use Lemma 1.1 to prove the above theorem. We shall
write each measure in the form of generating function, and then give their
first and second order derivatives. According to Lemma 1.1 it is sufficient
to show that the second order derivative in each case is nonnegative.

\newpage
\noindent
\textbf{$\bullet$ For $M_{SL} (a,b)$}: We can write

\[
M_{SL} (a,b) = a\;f_{SL} \left( {\frac{b}{a}} \right),
\quad
a \ne b,\;a,b > 0
\]

\noindent where

\[
f_{SL} (x) = \sqrt {\frac{x^2 + 1}{2}} - \frac{x - 1}{\ln x},
\quad
x \ne 1.
\]

\bigskip
Writing the first and second order derivatives of $f_{SL} (x)$, we have

\[
{f}'_{_{SL} } (x) = \frac{x}{\sqrt {2(x^2 + 1)} } - \frac{x\ln x - x +
1}{x\left( {\ln x} \right)^2},
\quad
x \ne 1
\]

\noindent and
\begin{align}
\label{eq24}
{f}''_{SL} (x) & = \frac{\left( {2x^2 + 2} \right)^{3 / 2}\left[ {(x + 1)\ln x
- 2(x - 1)} \right] + 2x^2\left( {\ln x} \right)^3}{x^2\left( {2x^2 + 2}
\right)^{3 / 2}\left( {\ln x} \right)^3}\notag\\
& = \frac{2}{\left( {2x^2 + 2} \right)^{3 / 2}} + \frac{2}{x^2\left( {\ln x}
\right)^2}\left( {\frac{x + 1}{2} - \frac{x - 1}{\ln x}} \right) \notag\\
& = \frac{2}{\left( {2x^2 + 2} \right)^{3 / 2}} + k(x),\; x \ne 1,
\end{align}

\noindent where
\begin{equation}
\label{eq25}
k(x) = \frac{2}{x^2\left( {\ln x} \right)^2}\left( {\frac{x + 1}{2} -
\frac{x - 1}{\ln x}} \right) = \frac{2}{x^2\left( {\ln x} \right)^2}f_{AL}
(x),\;x \ne 1,\;x > 0.
\end{equation}

We can easily check that
\begin{equation}
\label{eq26}
k(1) = \mathop {\lim }\limits_{x \to 1} k(x) = \frac{1}{6}.
\end{equation}

Since $A(a,b) \ge L(a,b)$, this gives that ${f}''_{AL} (x) \ge 0$, $\forall
x \in (0,\infty ),\;x \ne 1$. Moreover, $\mathop {\lim }\limits_{x \to 1}
f_{AL} (1) = \mathop {\lim }\limits_{x \to 1} {f}'_{AL} (1) = 0$.

\bigskip
\noindent
\textbf{$\bullet$ For $M_{AL} (a,b)$:}  We can write

\[
M_{AL} (a,b) = a\;f_{AL} \left( {\frac{b}{a}} \right),
\quad
a \ne b,\;a,b > 0
\]

\noindent where

\[
f_{AL} (x) = \frac{x + 1}{2} - \frac{x - 1}{\ln x},
\quad
x \ne 1
\]

\bigskip
Writing the first and second order derivatives of $f_{AL} (x)$, we have

\[
{f}'_{AL} (x) = \frac{x\ln x\left( {\ln x - 2} \right) + 2(x - 1)}{2x\left(
{\ln x} \right)^2},
\quad
x \ne 1,
\]

\noindent and
\begin{equation}
\label{eq27}
{f}''_{AL} (x) = \frac{(x + 1)\ln x - 2(x - 1)}{x^2\left( {\ln x} \right)^3}
= k(x),
\quad
x \ne 1,\;x > 0,
\end{equation}

Since $A(a,b) \ge L(a,b)$, this gives that ${f}''_{AL} (x) \ge 0$, $\forall
x \in (0,\infty ),\;x \ne 1$. Moreover, $\mathop {\lim }\limits_{x \to 1}
f_{AL} (1) = \mathop {\lim }\limits_{x \to 1} {f}'_{AL} (1) = 0$.

\bigskip
\noindent
\textbf{$\bullet$ For $M_{N_2 L} (a,b)$:}  We can write

\[
M_{N_2 L} (a,b) = a\;f_{N_2 L} \left( {\frac{b}{a}} \right),
\quad
a \ne b,\;a,b > 0
\]

\noindent where

\[
f_{N_2 L} (x) = \frac{x + \sqrt x + 1}{3} - \frac{x - 1}{\ln x},
\quad
x \ne 1
\]

\bigskip
Writing the first and second order derivatives of $f_{N_2 L} (x)$, we have

\[
{f}'_{N_2 L} (x) = \frac{\sqrt x \left( {2\sqrt x + 1} \right)\left( {\ln x}
\right)^2 - 6\left( {x\ln x - (x - 1)} \right)}{6x\left( {\ln x}
\right)^2},
\]

\noindent and
\begin{align}
\label{eq28}
{f}''_{N_2 L} (x) & = \frac{12\left( {(x + 1)\ln x - 2(x - 1)} \right) - \sqrt
x \left( {\ln x} \right)^3}{12x^2\left( {\ln x} \right)^3}. \notag\\
& = \frac{1}{12x^2(\ln x)^2}\left[ { - \sqrt x \left( {\ln x} \right)^2 +
24\left( {\frac{x + 1}{2} - \frac{x - 1}{\ln x}} \right)} \right] \notag\\
& = - \frac{1}{12x^{3 / 2}} + \frac{2}{x^2(\ln x)^2}\left( {\frac{x + 1}{2} -
\frac{x - 1}{\ln x}} \right) \notag\\
&= - \frac{1}{12x^{3 / 2}} + k(x).
 \end{align}

The above graph of the function ${f}''_{N_2 L} (x)$ is given by

\begin{center}
\includegraphics[bb=0mm 0mm 208mm 296mm, width=45.6mm, height=45.6mm, viewport=3mm 4mm 205mm 292mm]{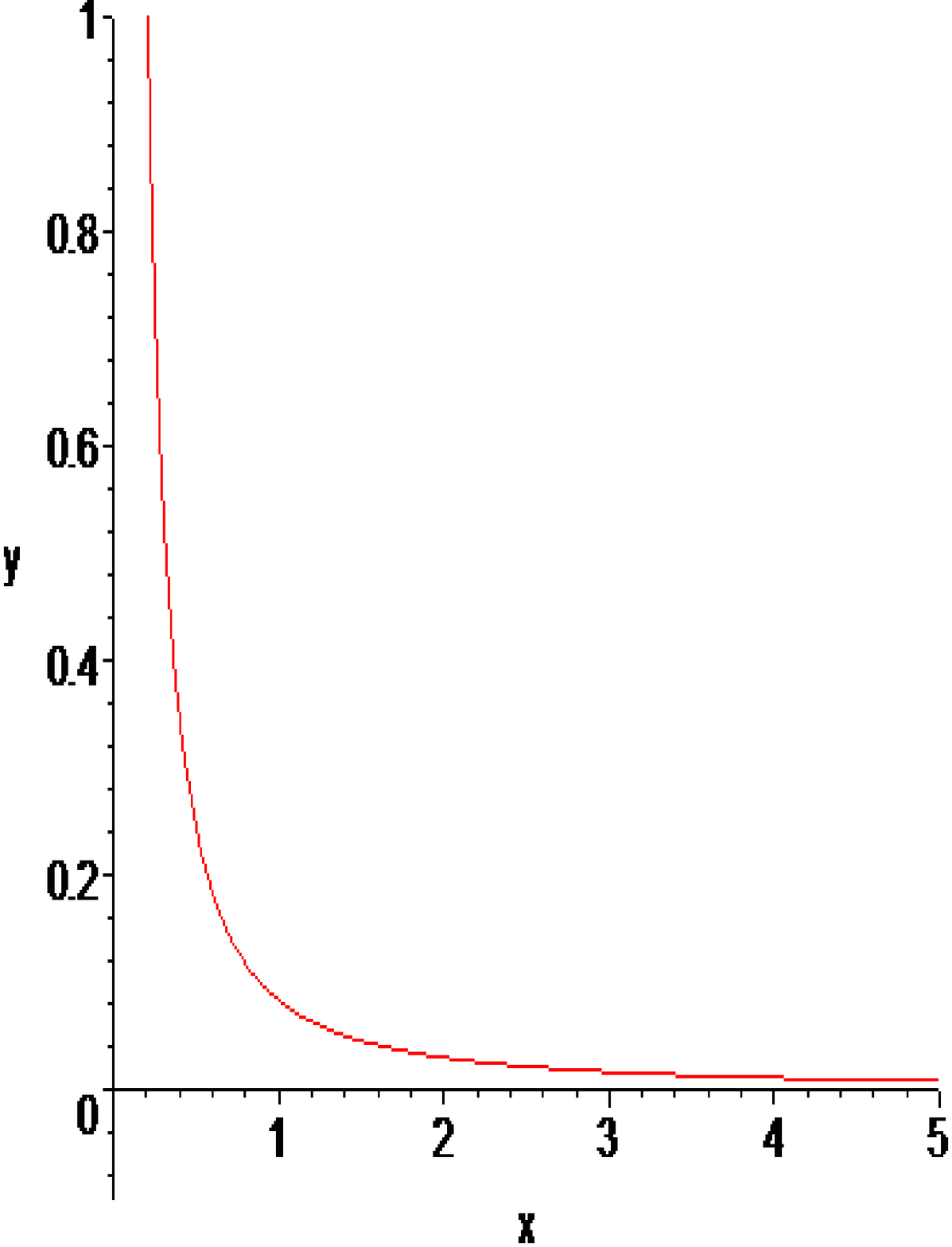}
\end{center}

From the above graph we observe that the function $k_1 (x)$ is nonnegative
for all $x \in (0,\infty )$, and consequently, ${f}''_{N_2 L} (x) \ge 0$,
$\forall x \in (0,\infty ),\;x \ne 1$.

\bigskip
Moreover, $\mathop {\lim }\limits_{x
\to 1} f_{N_2 L} (1) = \mathop {\lim }\limits_{x \to 1} {f}'_{N_2 L} (1) =
0$. Also we have

\[
\mathop {\lim }\limits_{x \to \infty } {f}''_{N_2 L} (x) = \mathop {\lim
}\limits_{x \to \infty } \left( { - \frac{1}{12x^{3 / 2}} + k(x)} \right) =
\infty
\]

\bigskip
\noindent
\textbf{$\bullet$ For $M_{N_3 L} (a,b)$:} We can write

\[
M_{N_3 L} (a,b) = a\;f_{N_3 L} \left( {\frac{b}{a}} \right),
\quad
a \ne b,\; a,b > 0
\]

\noindent where

\[
f_{N_3 L} (x) = \frac{\left( {\sqrt x + 1} \right)\sqrt {2(x + 1)} }{4} -
\frac{x - 1}{\ln x},
\quad
x \ne 1
\]

\bigskip
Writing the first and second order derivatives of $f_{N_3 L} (x)$, we have

\[
{f}'_{N_3 L} (x) = \frac{\sqrt x \left( {\ln x} \right)^2\left( {2x + \sqrt
x + 1} \right) - 4\sqrt {2(x + 1)} \left( {x\ln x - x + 1} \right)}{4x\left(
{\ln x} \right)^2\sqrt {2(x + 1)} },
\]

\noindent and
\begin{align}
\label{eq28a}
{f}''_{N_3 L} (x) & = \frac{4(2x + 2)^{3 / 2}\left( {(x + 1)\ln x - 2(x - 1)}
\right) - \sqrt x \left( {x^{3 / 2} + 1} \right)\left( {\ln x}
\right)^3}{4x^2\left( {\ln x} \right)^3(2x + 2)^{3 / 2}} \notag\\
& = \frac{1}{8x^2\left( {\ln x} \right)^2}\left[ { - \sqrt x \left(
{\frac{x^{3 / 2} + 1}{(2x + 2)^{3 / 2}}} \right)\left( {\ln x} \right)^2 +
16\left( {\frac{x + 1}{2} - \frac{x - 1}{\ln x}} \right)} \right] \notag\\
& = - \frac{(x^{3 / 2} + 1)}{8x^{3 / 2}(2x + 2)^{3 / 2} }+ k(x)
\end{align}

The graph of the function ${f}''_{N_3 L} (x)$ is given by

\begin{center}
\includegraphics[bb=0mm 0mm 208mm 296mm, width=45.6mm, height=45.6mm, viewport=3mm 4mm 205mm 292mm]{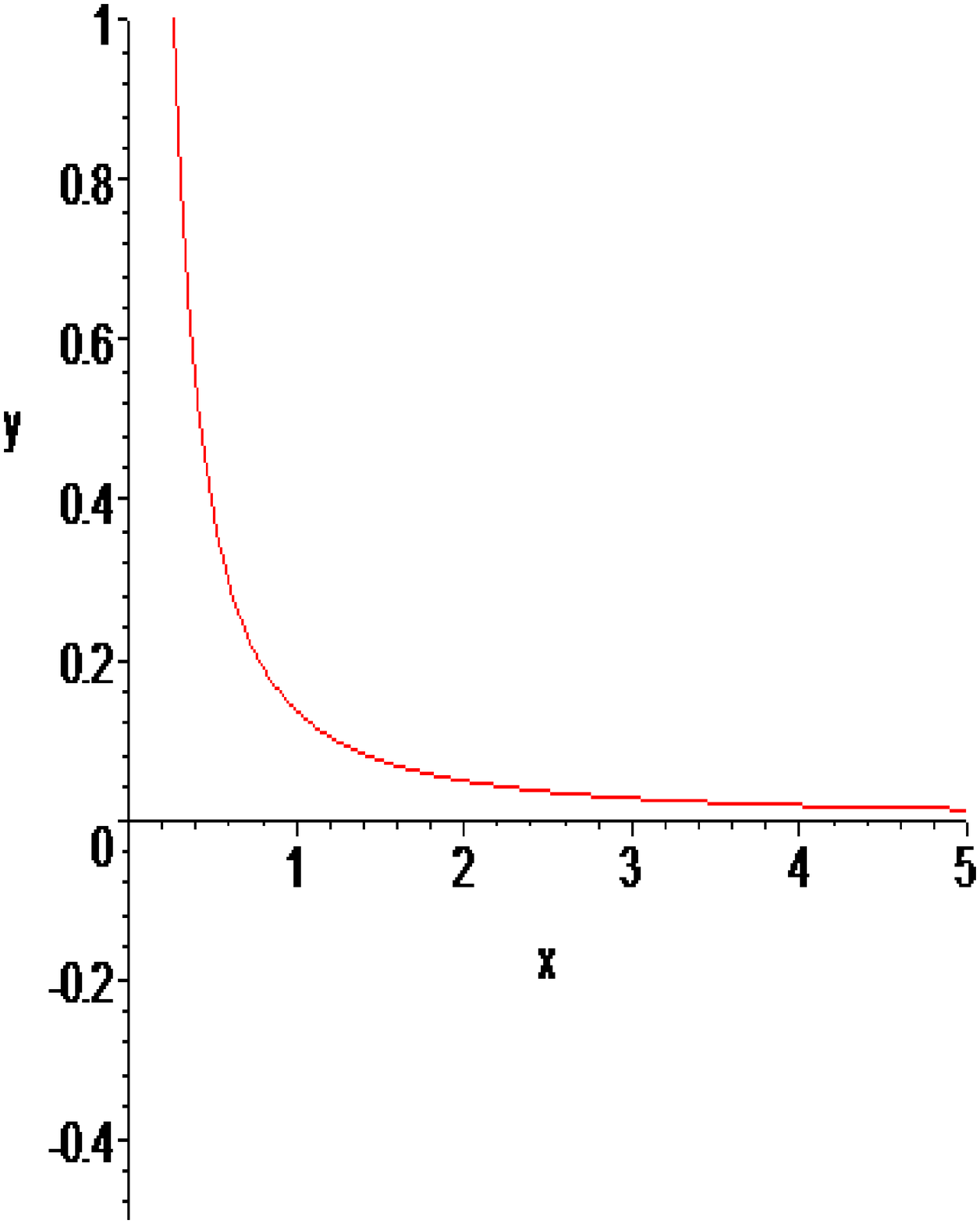}
\end{center}

From the above graph we observe that the function ${f}''_{N_3 L} (x)$ is
nonnegative for all $x \in (0,\infty )$, and consequently, ${f}''_{N_3 L}
(x) \ge 0$, $\forall x \in (0,\infty )$.

\bigskip
Moreover, $\mathop {\lim }\limits_{x \to 1} f_{N_3 L} (1) = \mathop {\lim
}\limits_{x \to 1} {f}'_{N_3 L} (1) = 0$. Also we have

\[
\mathop {\lim }\limits_{x \to \infty } {f}''_{N_3 L} (x) = \mathop {\lim
}\limits_{x \to \infty } \left( { - \frac{(x^{3 / 2} + 1)}{8x^{3 / 2}(2x +
2)^{3 /2} + k(x)} }\right) = \infty
\]

\bigskip
\noindent
\textbf{$\bullet$ For $M_{N_1 L} (a,b)$:} We can write

\[
M_{N_1 L} (a,b) = a\;f_{N_1 L} \left( {\frac{b}{a}} \right),
\quad
a \ne b,\;a,b > 0,
\]

\noindent where

\[
f_{N_1 L} (x) = \left( {\frac{\sqrt x + 1}{2}} \right)^2 - \frac{x - 1}{\ln
x},
\quad
x \ne 1
\]

\bigskip
Writing the first and second order derivatives of $f_{N_1 L} (x)$, we have

\[
{f}'_{_{N_1 L} } (x) = \frac{x\left( {\ln x} \right)^2(\sqrt x + 1) - 4\sqrt
x \left( {x\ln x - x + 1} \right)}{4x^{3 / 2}\left( {\ln x} \right)^2}
\]

\noindent and
\begin{align}
\label{eq29}
f''_{N_{1} L} (x) & =\frac{8\left[(x+1)\ln x-2(x-1)\right]-\sqrt{x} \left(\ln x\right)^{3} }{8x^{2} \left(\ln x\right)^{3} } \notag\\
&=\frac{1}{8x^{2} \left(\ln x\right)^{2} } \left[-\sqrt{x} \left(\ln x\right)^{2} +16\left(\frac{x+1}{2} -\frac{x-1}{\ln x} \right)\right]\notag\\
& =-\frac{1}{8x^{3/2} } +\frac{2}{x^{2} \left(\ln x\right)^{2} } \left(\frac{x+1}{2} -\frac{x-1}{\ln x} \right)=-\frac{1}{8x^{3/2} } +k(x)
\end{align}

The graph of the function ${f}''_{N_1 L} (x)$ is given by

\begin{center}
\includegraphics[bb=0mm 0mm 208mm 296mm, width=45.6mm, height=45.6mm, viewport=3mm 4mm 205mm 292mm]{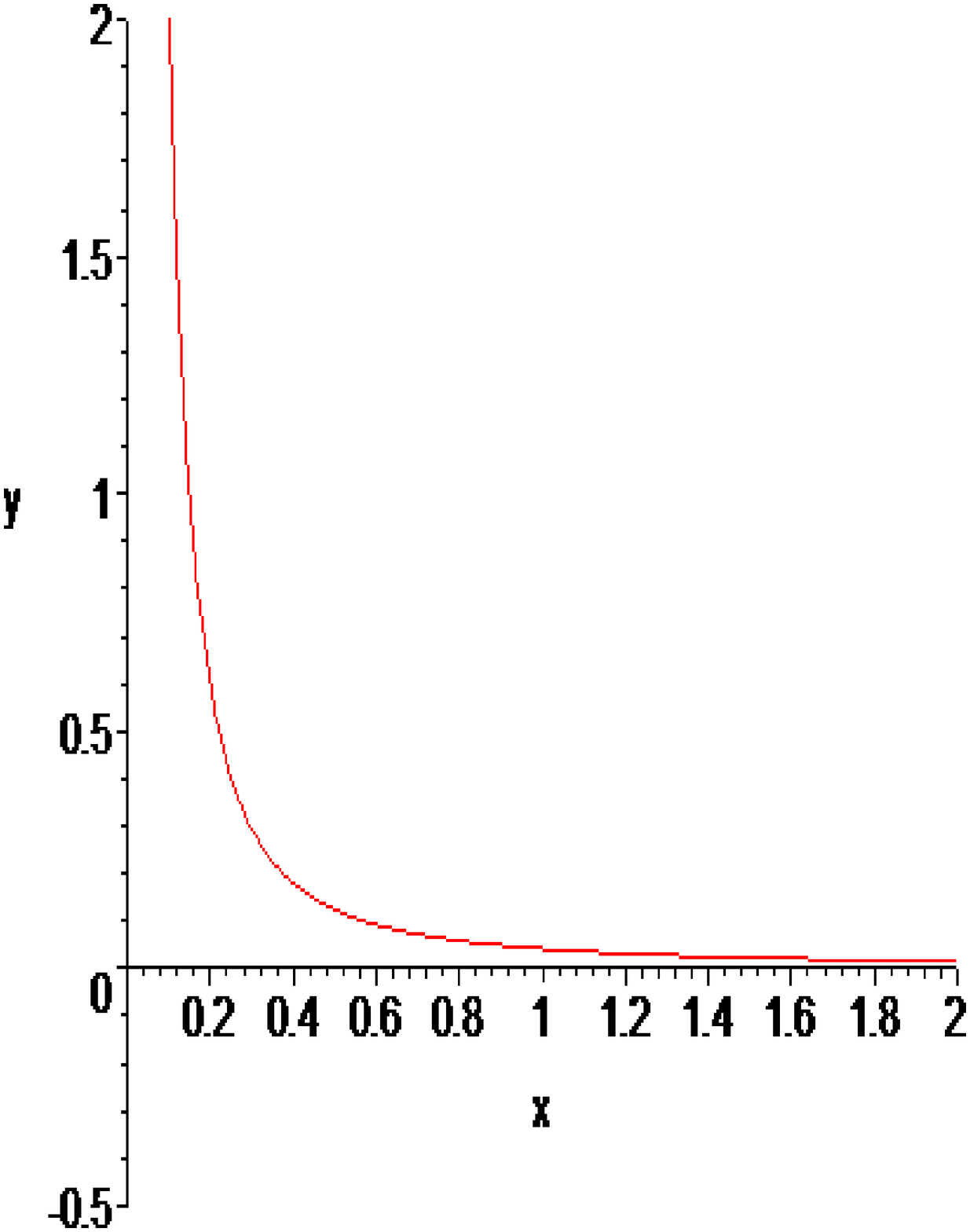}
\end{center}

From the above graph we observe that the function ${f}''_{N_1 L} (x)$ is
nonnegative for all $x \in (0,\infty )$, and consequently, ${f}''_{N_1 L}
(x) \ge 0$, $\forall x \in (0,\infty )$.

\bigskip
Moreover, $\mathop {\lim }\limits_{x \to 1} f_{N_1 L} (1) = \mathop {\lim
}\limits_{x \to 1} {f}'_{N_1 L} (1) = 0$. Also we have

\[
\mathop {\lim }\limits_{x \to \infty } {f}''_{N_1 L} (x) = \mathop {\lim
}\limits_{x \to \infty } \left( { - \frac{1}{8x^{3 / 2}} + k(x)} \right) =
\infty
\]

\bigskip
We see that in all the cases the generating function $f_{( \cdot )} (1) =
{f}'_{( \cdot )} (1) = 0$ and the second derivative is positive for all $x
\in (0,\infty )$. This proves the \textit{nonegativity} and \textit{convexity} of the means (\ref{eq19})-(\ref{eq23}) in ${\rm
R}_ + ^2 $. This completes the proof of the theorem.

\section{Log-Mean Inequalities}

In view of (\ref{eq18}), the following inequalities are obvious
\begin{align}
\label{eq31}
M_{N_1 L} (a,b)
& \le M_{N_3 L} (a,b) \le M_{N_2 L} (a,b) \le M_{AL} (a,b) \le M_{SL} (a,b),\\
\label{eq32}
M_{N_3 N_1 } (a,b) & \le M_{N_2 N_1 } (a,b) \le M_{AN_1 } (a,b) \le M_{SN_1 }
(a,b),\\
\label{eq33}
M_{N_2 N_3 } (a,b) & \le M_{AN_3 } (a,b) \le M_{SN_3 } (a,b)\\
\intertext{and}
\label{eq34}
M_{AN_3 } (a,b) & \le M_{SN_2 } (a,b).
\end{align}

Here we shall prove an improvement over the inequalities (\ref{eq31}). While, the
improvement over the inequalities (\ref{eq32})-(\ref{eq34}) is already given in
(\ref{eq9})-(\ref{eq11}). The following theorem holds:

\begin{theorem} The following inequalities hold:
\begin{equation}
\label{eq35}
M_{SL} (a,b) \le \frac{5}{2}M_{AL} (a,b) \le 5M_{N_2 L} (a,b) \le 6M_{N_1 L}
\end{equation}

\noindent and
\begin{equation}
\label{eq36}
M_{SL} (a,b) \le 4\;M_{N_3 L} (a,b) \le 10\;M_{N_1 L} .
\end{equation}
\end{theorem}

The proof of the above theorem is based on Lemma 1.2 and is given in parts
in the following propositions.

\begin{proposition}  We have
\begin{equation}
\label{eq37}
M_{SL} (a,b) \le \frac{5}{2}M_{AL} (a,b).
\end{equation}
\end{proposition}

\begin{proof} Let us consider
\[
g_{SL\_AL} (x) = \frac{{f}''_{SL} (x)}{{f}''_{AL} (x)} = \frac{2x^2(\ln x)^3
+ (2 + 2x^2)^{3 / 2}\left[ {(x + 1)\ln x - 2(x - 1)} \right]}{(2 + 2x^2)^{3
/ 2}\left[ {(x + 1)\ln x - 2(x - 1)} \right]},\,\;x \ne 1
\]

\noindent
for all $x \in (0,\infty )$, where ${f}''_{SL} (x)$ and ${f}''_{AL} (x)$ are
as given by (\ref{eq24}) and (\ref{eq27}) respectively.

\bigskip
Calculating the first order derivative of the function $g_{SL\_AL} (x)$ with
respect to $x$, one gets
\begin{align}
 {g}'_{SH\_SL} (x) =&  \frac{4x(\ln x)^2}{(2 + 2x^2)^{5 / 2}\left[ {(x + 1)\ln
x - 2(x - 1)} \right]^2}\times \notag\\
&\left\{ {6(x - 1)(x^2 + 1) - \ln x\left[ {6(x^3 + 1) - (x^2 - x - 2)\ln x}
\right] + 2x^3(\ln x)^2} \right\} \notag
\end{align}

\bigskip
The graph of the function ${g}'_{SH\_SL} (x)$ is given by

\begin{center}
\includegraphics[bb=0mm 0mm 208mm 296mm, width=45.6mm, height=45.6mm, viewport=3mm 4mm 205mm 292mm]{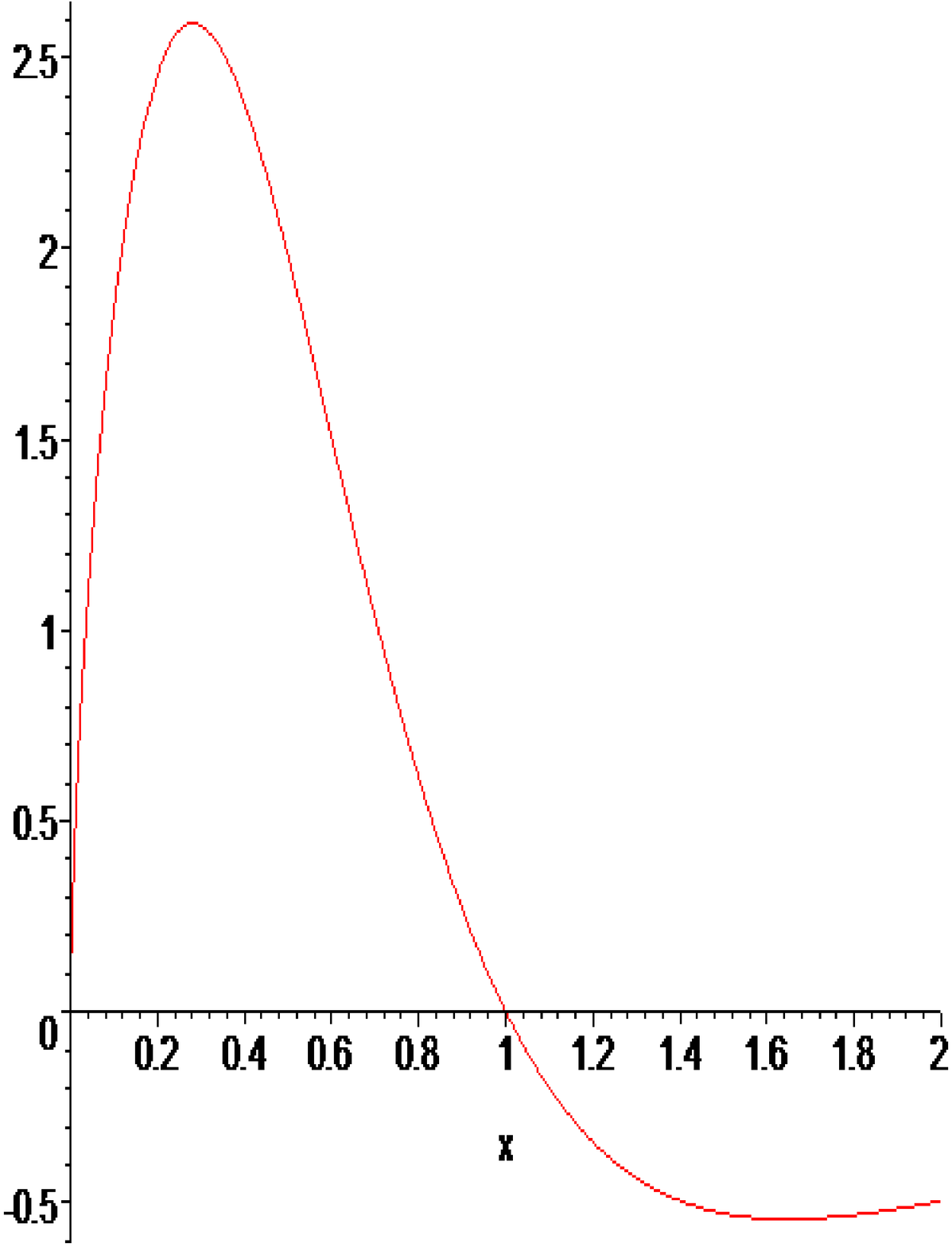}
\end{center}

\bigskip
We observe from the above graph the following:
\begin{equation}
\label{eq38}
{g}'_{SL\_AL} (x) = \left\{ {{\begin{array}{*{20}c}
 { > 0,} & {x < 1} \\
 { < 0,} & {x > 1} \\
\end{array} }} \right.
\end{equation}

Let us calculate now $g_{SL\_AL} (1)$. We observe that
\[
\left. {g_{SL\_AL} (x)} \right|_{x = 1} = \left. {\frac{\left( {{f}''_{SL}
(x)} \right)^\prime }{\left( {{f}''_{AL} (x)} \right)^\prime }} \right|_{x =
1} = \left. {\frac{\left( {{f}''_{SL} (x)} \right)^{\prime \prime }}{\left(
{{f}''_{AL} (x)} \right)^{\prime \prime }}} \right|_{x = 1} =
\mbox{indermination}.
\]

\bigskip
Calculating third order derivatives of numerator and denominator, we have
\begin{equation}
\label{eq39}
g_{SL\_AL} (1) = \left. {\frac{\left( {{f}''_{SL} (x)} \right)^{\prime
\prime \prime }}{\left( {{f}''_{AL} (x)} \right)^{\prime \prime \prime }}}
\right|_{x = 1} = \frac{10}{4} = \frac{5}{2}.
\end{equation}

By the application of (\ref{eq14}) with (\ref{eq39}) we get (\ref{eq37}).
\end{proof}
\begin{proposition}  We have
\begin{equation}
\label{eq40}
M_{AL} (a,b) \le 2M_{N_2 L} (a,b).
\end{equation}
\end{proposition}

\begin{proof}
Let us consider
\[
g_{AL\_N_2 L} (x) = \frac{{f}''_{AL} (x)}{{f}''_{N_2 L} (x)} = \frac{12x^{3
/ 2}\left[ {(x + 1)\ln x - 2(x - 1)} \right]}{x^2(\ln x)^3 - 12x^{3 /
2}\left[ {(x + 1)\ln x - 2(x - 1)} \right]},
\quad
x \ne 1
\]

\noindent for all $x \in (0,\infty )$, where ${f}''_{AL} (x)$ and ${f}''_{N_2 L} (x)$
are as given by (\ref{eq27}) and (\ref{eq28}) respectively.

\bigskip
Calculating the first order derivative of the function $g_{AL\_N_2 L} (x)$
with respect to $x$, one gets
\[
{g}'_{AL\_N_2 L} (x) = \frac{6x^{5 / 2}(\ln x)^2\left\{ {\ln x\left[ {(x -
1)\ln x - 6(x + 1)} \right] - 12(x - 1)} \right\}}{\left\{ {x^2(\ln x)^3 -
12x^{3 / 2}\left[ {(x + 1)\ln x - 2(x - 1)} \right]} \right\}^2},\;x \ne 1
\]

The graph of the function ${g}'_{AL\_N_2 L} (x)$ is given by

\begin{center}
\includegraphics[bb=0mm 0mm 208mm 296mm, width=45.6mm, height=45.6mm, viewport=3mm 4mm 205mm 292mm]{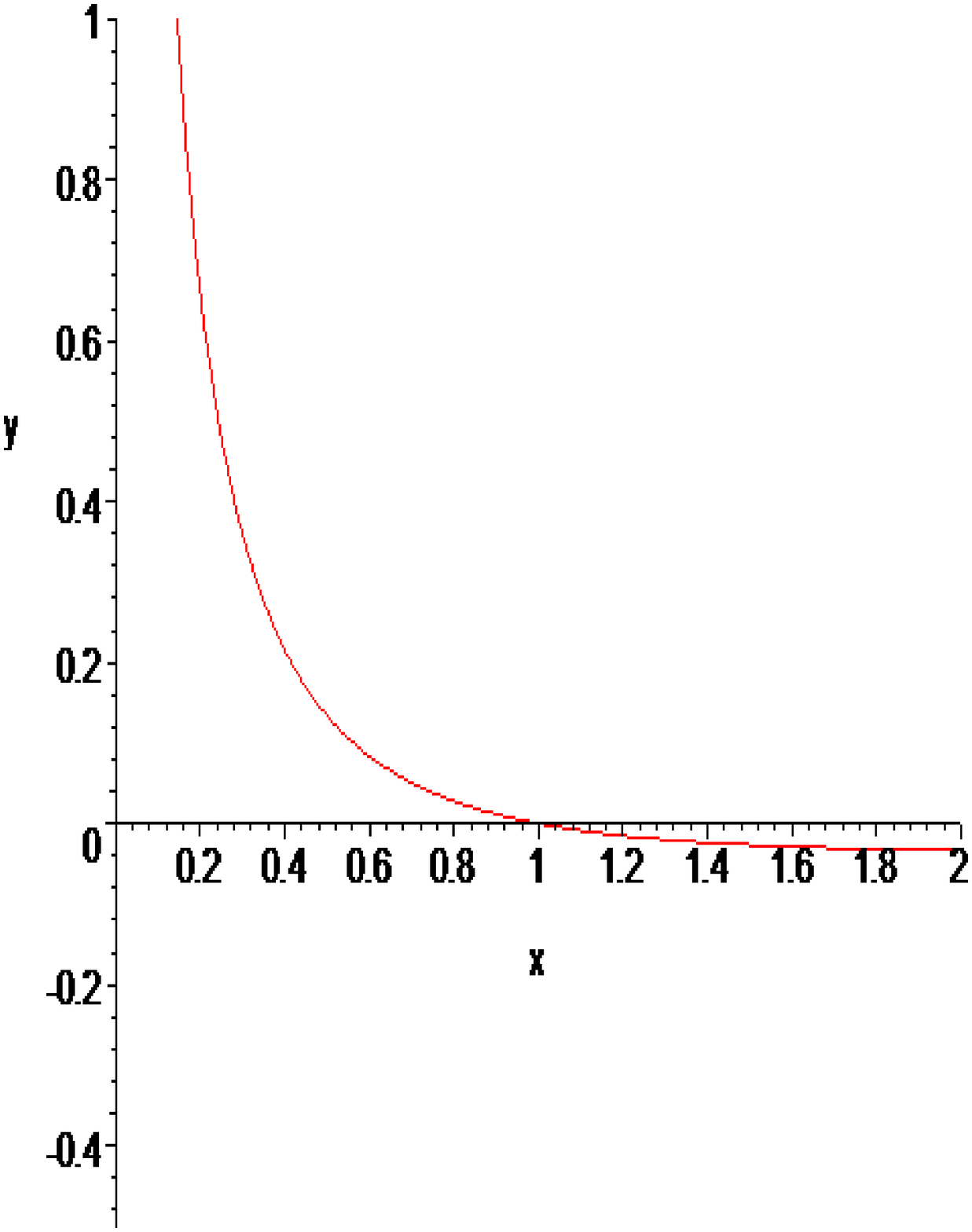}
\end{center}

We observe from the above graph the following:
\begin{equation}
\label{eq41}
{g}'_{AL\_N_2 L} (x) = \left\{ {{\begin{array}{*{20}c}
 { > 0,} & {x < 1} \\
 { < 0,} & {x > 1} \\
\end{array} }} \right.
\end{equation}

Let us calculate now $g_{AL\_N_2 L} (1)$. We observe that

\[
\left. {g_{AL\_N_2 L} (x)} \right|_{x = 1} = \left. {\frac{\left(
{{f}''_{AL} (x)} \right)^\prime }{\left( {{f}''_{N_2 L} (x)} \right)^\prime
}} \right|_{x = 1} = \left. {\frac{\left( {{f}''_{AL} (x)} \right)^{\prime
\prime }}{\left( {{f}''_{N_2 L} (x)} \right)^{\prime \prime }}} \right|_{x =
1} = \mbox{indermination}
\]

Calculating third order derivatives of numerator and denominator, we have
\begin{equation}
\label{eq42}
g_{AL\_N_2 L} (1) = \left. {\frac{\left( {{f}''_{AL} (x)} \right)^{\prime
\prime \prime }}{\left( {{f}''_{N_2 L} (x)} \right)^{\prime \prime \prime
}}} \right|_{x = 1} = \frac{ - 12}{ - 6} = 2.
\end{equation}

By the application of (\ref{eq14}) with (\ref{eq42}) we get (\ref{eq40}).
\end{proof}

\begin{proposition}  We have
\begin{equation}
\label{eq43}
M_{N_2 L} (a,b) \le 2M_{N_1 L} (a,b).
\end{equation}
\end{proposition}

\begin{proof}  Let us consider
\[
g_{N_2 L\_N_1 L} (x) = \frac{{f}''_{N_2 L} (x)}{{f}''_{N_1 L} (x)} =
\frac{2}{3}\frac{x^2(\ln x)^3 + 12x^{3 / 2}\left[ {2(x - 1) - (x + 1)\ln x}
\right]}{x^2(\ln x)^3 + 8x^{3 / 2}\left[ {2(x - 1) - (x + 1)\ln x} \right]},
\quad
x \ne 1
\]

\noindent for all $x \in (0,\infty )$, where ${f}''_{N_2 L} (x)$ and ${f}''_{N_1 L}
(x)$ are as given by (\ref{eq28}) and (\ref{eq29}) respectively.

\bigskip
Calculating the first order derivative of the function $g_{N_2 L\_N_1 L}
(x)$ with respect to $x$, one gets

\[
{g}'_{N_2 L\_N_1 L} (x) = - \frac{4}{3}\frac{x^{5 / 2}(\ln x)^2\left\{ {12(x
- 1) + \ln x\left[ {(x - 1)\ln x - 6(x + 1)} \right]} \right\}}{\left\{
{x^2(\ln x)^3 + 8x^{3 / 2}\left[ {2(x - 1) - (x + 1)\ln x} \right]}
\right\}^2}.
\]

The graph of the function ${g}'_{N_2 L\_N_1 L} (x)$ is given by

\begin{center}
\includegraphics[bb=0mm 0mm 208mm 296mm, width=45.6mm, height=45.6mm, viewport=3mm 4mm 205mm 292mm]{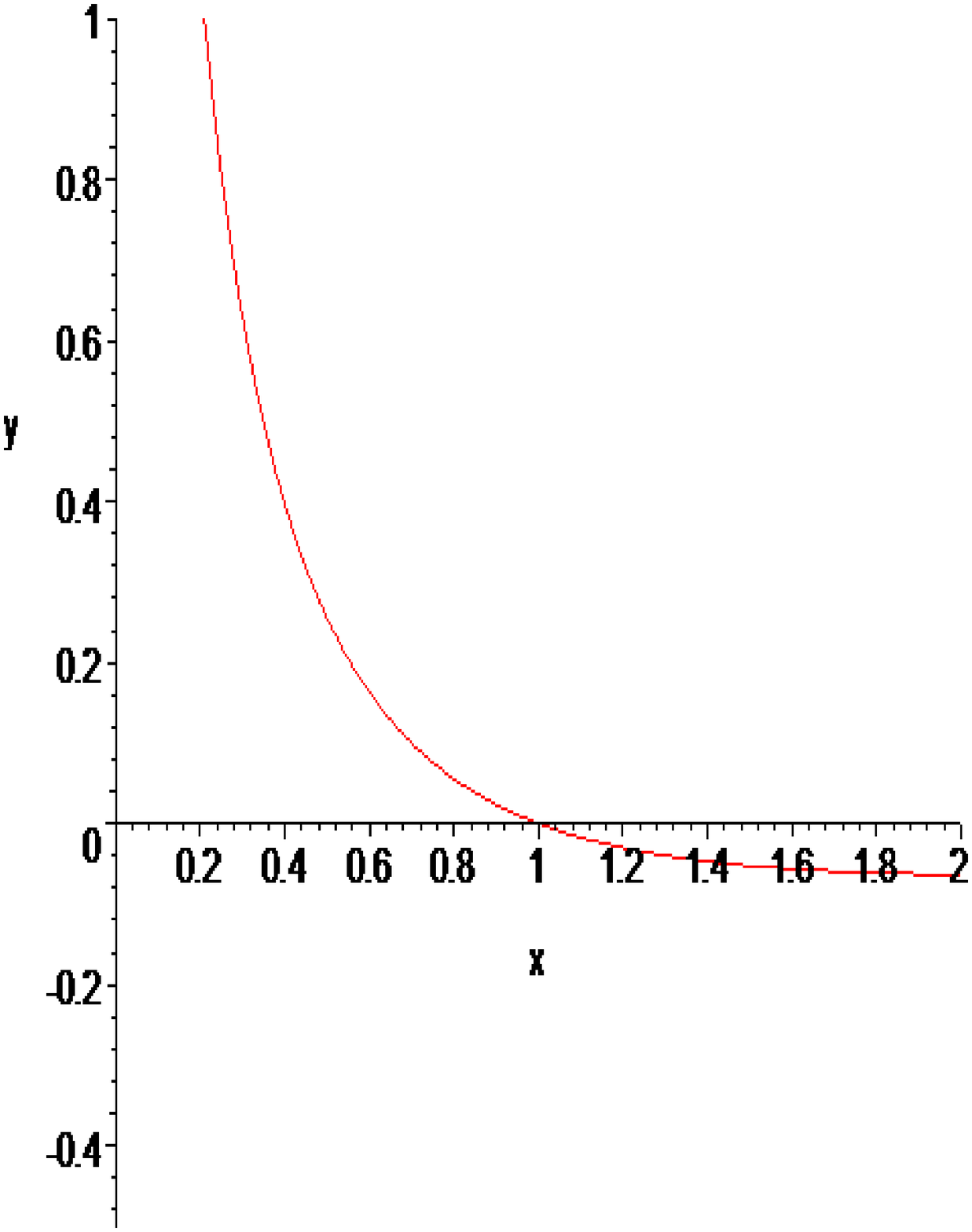}
\end{center}

We observe from the above graph the following:
\begin{equation}
\label{eq44}
{g}'_{N_2 L\_N_1 L} (x) = \left\{ {{\begin{array}{*{20}c}
 { > 0,} & {x < 1} \\
 { < 0,} & {x > 1} \\
\end{array} }} \right.
\end{equation}

Let us calculate now $g_{N_2 L\_N_1 L} (1)$. We observe that
\[
\left. {g_{N_2 L\_N_1 L} (x)} \right|_{x = 1} = \left. {\frac{\left(
{{f}''_{N_2 L} (x)} \right)^\prime }{\left( {{f}''_{N_1 L} (x)}
\right)^\prime }} \right|_{x = 1} = \left. {\frac{\left( {{f}''_{N_2 L} (x)}
\right)^{\prime \prime }}{\left( {{f}''_{N_1 L} (x)} \right)^{\prime \prime
}}} \right|_{x = 1} = \mbox{indermination}
\]

Calculating third order derivatives of numerator and denominator, we have
\begin{equation}
\label{eq45}
g_{N_2 L\_N_1 L} (1) = \left. {\frac{\left( {{f}''_{N_2 L} (x)}
\right)^{\prime \prime \prime }}{\left( {{f}''_{N_1 L} (x)} \right)^{\prime
\prime \prime }}} \right|_{x = 1} = \frac{ - 12}{ - 6} = 2.
\end{equation}

By the application of (\ref{eq14}) with (\ref{eq45}) we get (\ref{eq43}).
\end{proof}

\begin{proposition} We have
\begin{equation}
\label{eq46}
M_{SL} (a,b) \le 4\,M_{N_3 L} (a,b).
\end{equation}
\end{proposition}

\begin{proof} Let us consider
\begin{align}
g_{SL\_N_3 L} (x) & = \frac{{f}''_{SL} (x)}{{f}''_{N_3 L} (x)}\notag\\
& = \frac{ - 4(2x + 2)^{3 / 2}\left[ {2x^2(\ln x)^3 + (2x^2 + 2)^{3 /
2}\left( {(x + 1)\ln x - 2(x - 1)} \right)} \right]}{(2x^2 + 2)^{3 /
2}\left[ {(x^{3 / 2} + 1)x\ln x - 4(2x + 2)^{3 / 2}\left( {(x + 1)\ln x -
2(x - 1)} \right)} \right]},\;x \ne 1.\notag
\end{align}

\noindent for all $x \in (0,\infty )$, where ${f}''_{SL} (x)$ and ${f}''_{N_3 L} (x)$
are as given by (\ref{eq24}) and (\ref{eq28a}) respectively.

\bigskip
Calculating the first order derivative of the function $g_{SL\_N_3 L} (x)$
with respect to $x$, and making necessary calculations the graph of the
function ${g}'_{SL\_N_3 L} (x)$ is given by

\begin{center}
\includegraphics[bb=0mm 0mm 208mm 296mm, width=45.6mm, height=45.6mm, viewport=3mm 4mm 205mm 292mm]{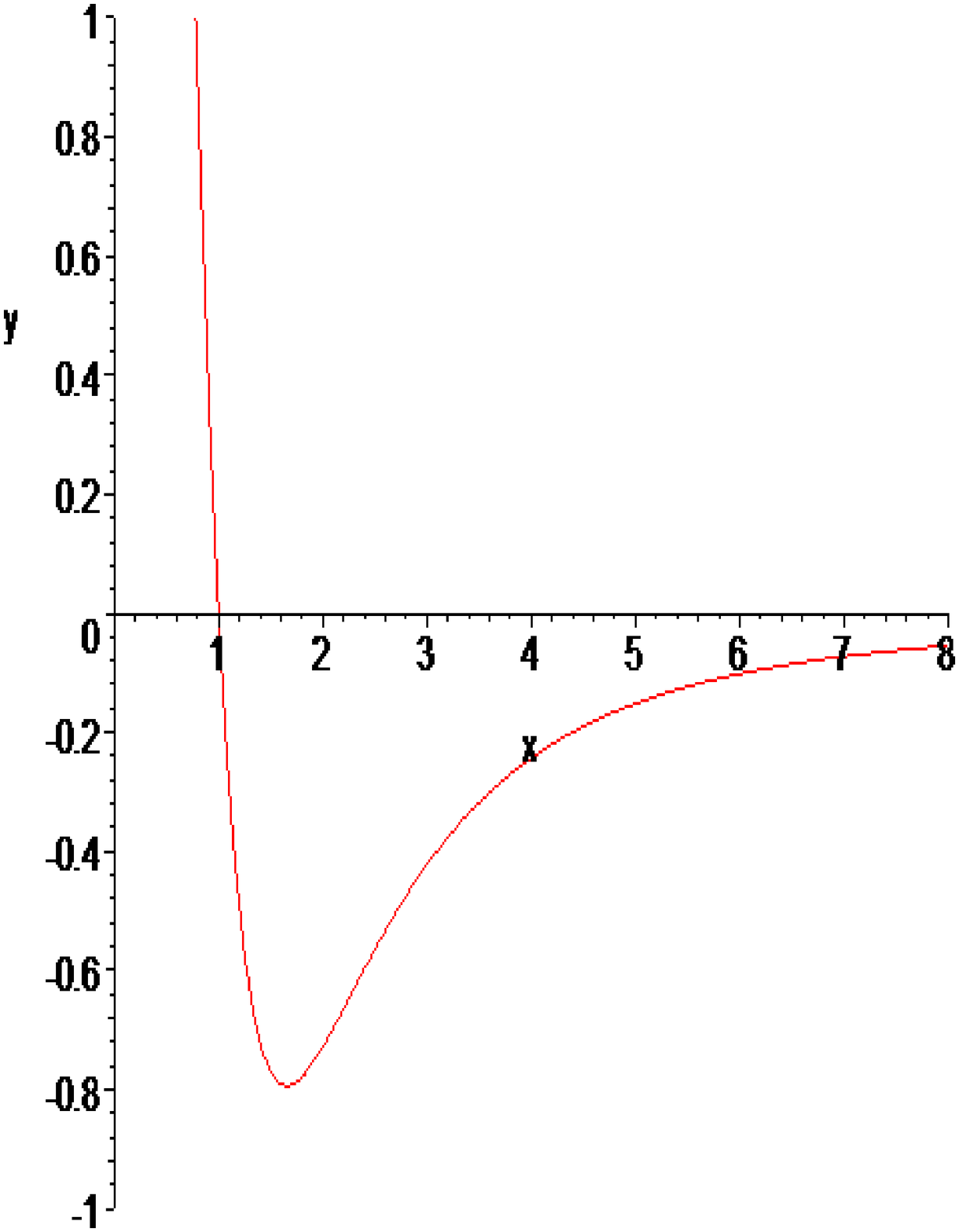}
\end{center}

We observe from the above graph the following:
\begin{equation}
\label{eq47}
{g}'_{SL\_N_3 L} (x) = \left\{ {{\begin{array}{*{20}c}
 { > 0,} & {x < 1} \\
 { < 0,} & {x > 1} \\
\end{array} }} \right.
\end{equation}

Let us calculate now $g_{SL\_N_3 L} (1)$. We observe that

\[
\left. {g_{SL\_N_3 L} (x)} \right|_{x = 1} = \left. {\frac{\left(
{{f}''_{SL} (x)} \right)^\prime }{\left( {{f}''_{N_3 L} (x)} \right)^\prime
}} \right|_{x = 1} = \left. {\frac{\left( {{f}''_{SL} (x)} \right)^{\prime
\prime }}{\left( {{f}''_{N_3 L} (x)} \right)^{\prime \prime }}} \right|_{x =
1} = \mbox{indermination}
\]

\bigskip
Calculating third order derivatives of numerator and denominator, we have
\begin{equation}
\label{eq48}
\left. {g_{SL\_N_3 L} (x)} \right|_{x = 1} = \left. {\frac{\left(
{{f}''_{SL} (x)} \right)^{\prime \prime \prime }}{\left( {{f}''_{N_3 L} (x)}
\right)^{\prime \prime \prime }}} \right|_{x = 1} = \frac{ - 160\sqrt 2 }{ -
40\sqrt 2 } = 4
\end{equation}

By the application of (\ref{eq14}) with (\ref{eq48}) we get (\ref{eq46}).
\end{proof}

\begin{proposition}  We have
\begin{equation}
\label{eq49}
M_{N_3 L} (a,b) \le \frac{5}{2}M_{N_1 L} (a,b).
\end{equation}
\end{proposition}

\begin{proof} Let us consider
\begin{align}
g_{N_3 L\_N_1 L} (x) & = \frac{{f}''_{N_3 L} (x)}{{f}''_{N_1 L} (x)}\notag\\
&  = \frac{x^2(\ln x)^3(x^{3 / 2} + 1) - 4x^{3 / 2}(2x + 2)^{3 / 2}\left( {(x
+ 1)\ln x - 2(x - 1)} \right)}{(2x + 2)^{3 / 2}\left[ { - x^2(\ln x)^3 +
8x^{3 / 2}\left( {(x + 1)\ln x - 2(x - 1)} \right)} \right]},\;x \ne 1.\notag
\end{align}

\noindent for all $x \in (0,\infty )$, where ${f}''_{N_3 L} (x)$ and ${f}''_{N_1 L}
(x)$ are as given by (\ref{eq28a}) and (\ref{eq29}) respectively.

\bigskip
Calculating the first order derivative of the function $g_{N_3 L\_N_1 L}
(x)$ with respect to $x$ and making necessary calculations the graph of the
function ${g}'_{N_3 L\_N_1 L} (x)$ is given by

\begin{center}
\includegraphics[bb=0mm 0mm 208mm 296mm, width=45.6mm, height=45.6mm, viewport=3mm 4mm 205mm 292mm]{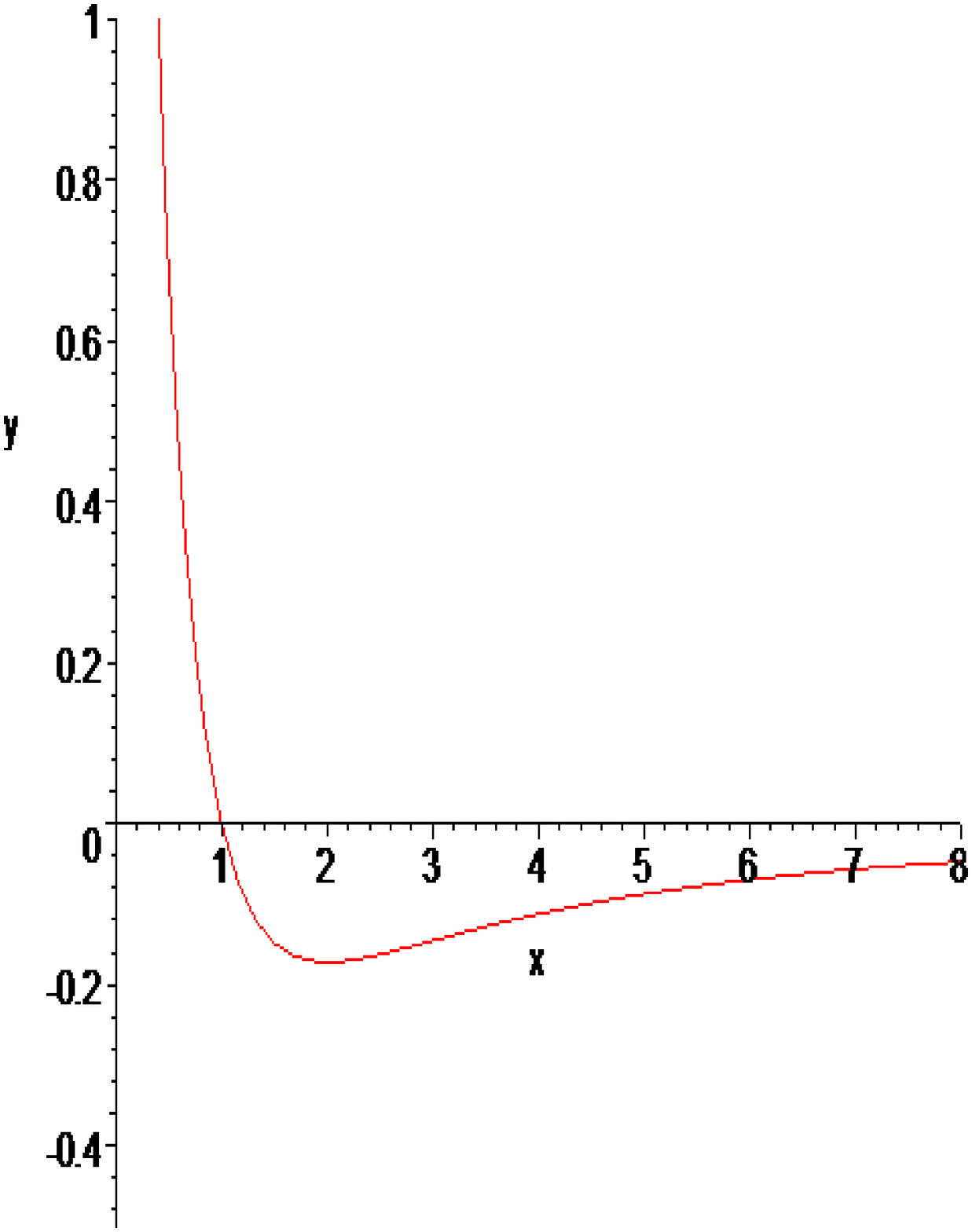}
\end{center}

We observe from the above graph the following:
\begin{equation}
\label{eq50}
{g}'_{N_3 L\_N_1 L} (x) = \left\{ {{\begin{array}{*{20}c}
 { > 0,} & {x < 1} \\
 { < 0,} & {x > 1} \\
\end{array} }} \right.
\end{equation}

Let us calculate now $g_{N_3 L\_N_1 L} (1)$. We observe that

\[
\left. {g_{N_3 L\_N_1 L} (x)} \right|_{x = 1} = \left. {\frac{\left(
{{f}''_{N_3 L} (x)} \right)^\prime }{\left( {{f}''_{N_1 L} (x)}
\right)^\prime }} \right|_{x = 1} = \left. {\frac{\left( {{f}''_{N_3 L} (x)}
\right)^{\prime \prime }}{\left( {{f}''_{N_1 L} (x)} \right)^{\prime \prime
}}} \right|_{x = 1} = \mbox{indermination}
\]

\bigskip
Calculating third order derivatives of numerator and denominator, we have
\begin{equation}
\label{eq51a}
\left. {g_{N_3 L\_N_1 L} (x)} \right|_{x = 1} = \left. {\frac{\left(
{{f}''_{N_3 L} (x)} \right)^{\prime \prime \prime }}{\left( {{f}''_{N_1 L}
(x)} \right)^{\prime \prime \prime }}} \right|_{x = 1} = \frac{ - 20}{ - 8}
= \frac{5}{2}.
\end{equation}

By the application of (\ref{eq14}) with (\ref{eq51a}) we get (\ref{eq49}).
\end{proof}

\bigskip
\textbf{Proof of the theorem 3.2.} Propositions 3.1-3.3 together proves the
inequalities (\ref{eq35}) and the propositions 3.4-3.5 proves the inequalities
(\ref{eq36}).

\section{Remarks}

In this section we shall give some remarks based on the Theorem 3.1 given in
section 3.

\begin{remark} Following the similar lines of Theorem 3.1, still we
can prove that
\begin{equation}
\label{eq51}
M_{SH} (a,b)  \le \frac{9}{5}M_{SL} (a,b)
\end{equation}
\begin{equation}
\label{eq52}
M_{AG} (a,b) \le \frac{3}{2}M_{AL} (a,b)
\end{equation}
\noindent and
\begin{equation}
\label{eq53}
M_{SN_1 } (a,b)  \le \frac{9}{10}M_{SL} (a,b),
\end{equation}
\noindent where
\[M_{SH} (a,b) = S(a,b) - H(a,b) \]
\[M_{AG} (a,b) = A(a,b) - G(a,b) \]
\noindent and
\[M_{SN_1 } (a,b) = S(a,b) - N_1 (a,b).\]
\end{remark}

\bigskip
The convexity of the measures $M_{SH} (a,b)$, $M_{AG} (a,b)$ and $M_{SN_1 }
(a,b)$ can be seen in Taneja \cite{tan3}

\begin{remark} The following inequalities hold:
\begin{equation}
\label{eq54}
M_{SH} (a,b) \le \frac{9}{5}M_{SL} (a,b) \le 3M_{AG} (a,b) \le
\frac{9}{2}M_{AL} (a,b)
\end{equation}

\noindent and
\begin{equation}
\label{eq55}
M_{SA} (a,b) \le \frac{3}{4}M_{SN_3 } (a,b) \le \frac{2}{3}M_{SN_1 } (a,b)
\le \frac{3}{5}M_{SL} (a,b)
\end{equation}
\end{remark}

\begin{proof}  We know that the following inequalities hold:
\begin{equation}
\label{eq56}
M_{SA} (a,b) \le \frac{1}{3}M_{SH} (a,b) \le \frac{1}{2}M_{AH} (a,b) \le
\frac{1}{2}M_{SG} (a,b) \le M_{AG} (a,b).
\end{equation}

\noindent and
\begin{equation}
\label{eq57}
M_{SA} (a,b) \le \frac{3}{4}M_{SN_3 } (a,b) \le \frac{2}{3}M_{SN_1 } (a,b).
\end{equation}

The proof of the inequalities (\ref{eq56}) and (\ref{eq57}) can be seen in Taneja \cite{tan3}.

\bigskip
Inequalities (\ref{eq56}) together with (\ref{eq51}) and (\ref{eq52}) give (\ref{eq54}).
Inequalities (\ref{eq57}) together with (\ref{eq53}) give (\ref{eq55}). Still we need to show
that
\begin{equation}
\label{eq58}
\frac{9}{5}M_{SL} (a,b) \le 3M_{AG} (a,b),
\end{equation}

\noindent
i.e.,
\[
M_{SL} (a,b) \le \frac{5}{3}M_{AG} (a,b)
\]

\bigskip
In order to show this let us consider the difference
\begin{align}
T_1 (a,b)&  = \frac{5}{3}M_{AG} (a,b) - M_{SL} (a,b)\notag\\
 & = \frac{5A(a,b) - 5G(a,b)}{3} - S(a,b) + L(a,b) = a\;f_{T_1 } \left({\frac{b}{a}} \right),\notag
 \end{align}
\noindent where
\[f_{T_1 } (x)  = \frac{5(x + 1)}{6} - \frac{5\sqrt x }{3} - \frac{\sqrt
{2x^2 + 2} }{2} + \frac{x - 1}{\ln x}, \quad x \ne 1, \quad x > 0 \]

\bigskip
The graph of the function $f_{T_1 } (x)$ is given by
\begin{center}
\includegraphics[bb=0mm 0mm 208mm 296mm, width=45.6mm, height=45.6mm, viewport=3mm 4mm 205mm 292mm]{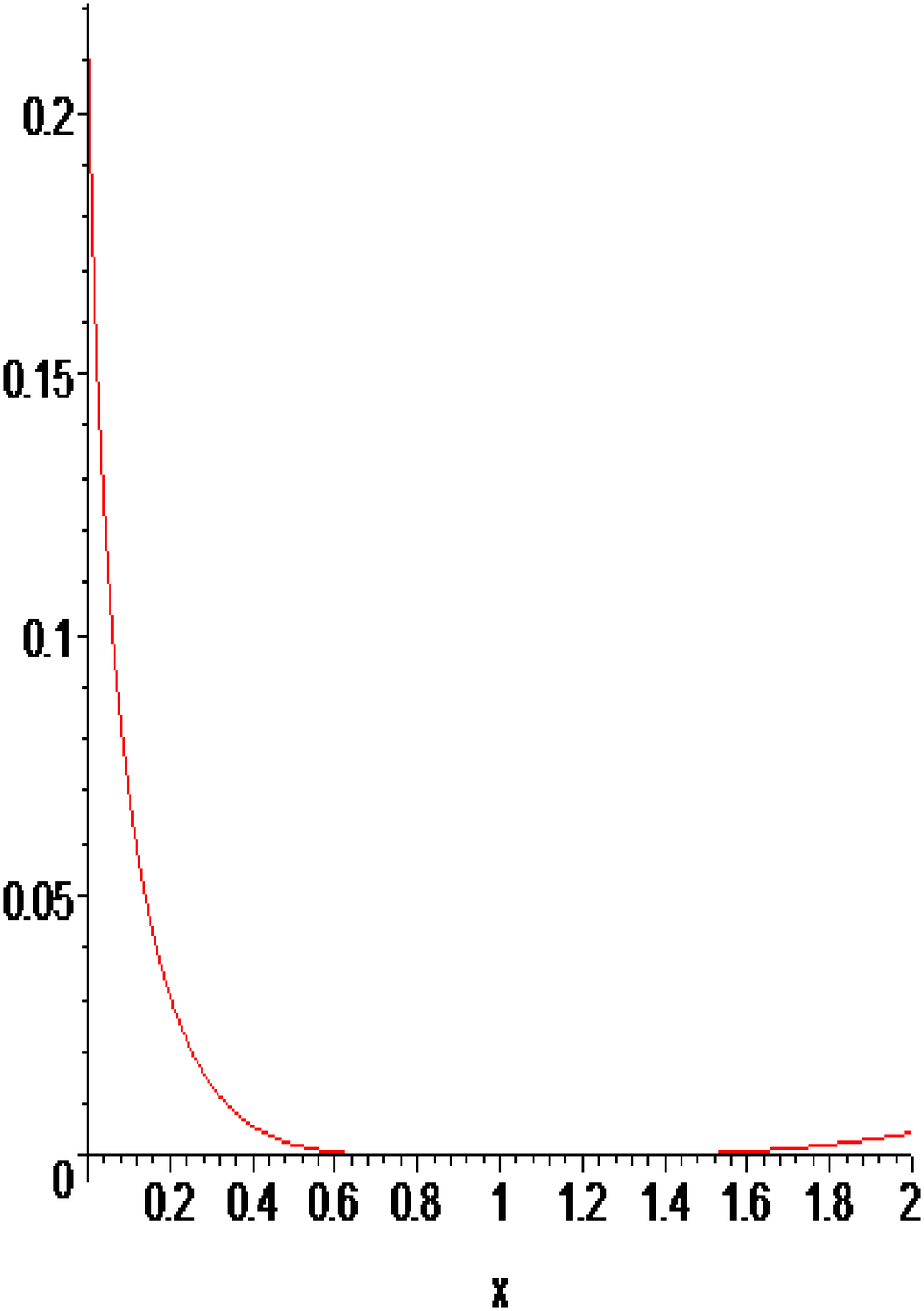}
\end{center}

From the above graph it is clear that $T_1 (a,b) \ge 0$. Proving the
inequalities (\ref{eq58}), consequently, proving (\ref{eq54}).
\end{proof}

\begin{remark} The following inequalities hold
\begin{align}
\label{eq59}
L(a,b) \le &  \frac{S(a,b) + 9L(a,b)}{10} \le \frac{2N_3 (a,b) + 3L(a,b)}{5}\notag\\
 & \le \frac{5A(a,b) + 7L(a,b)}{12} \le N_1 (a,b) \le \frac{5N_2 (a,b) + L(a,b)}{6}.
\end{align}
\end{remark}

\begin{proof} Simplifying the inequalities (\ref{eq34}) and (\ref{eq35}) given in
Theorem 3.1, we get
\begin{equation}
\label{eq60}
\frac{S(a,b) + 5L(a,b)}{6} \le \frac{5A(a,b) + 7L(a,b)}{12} \le N_1 (a,b)
\le \frac{5N_2 (a,b) + L(a,b)}{6}
\end{equation}

\noindent and
\begin{equation}
\label{eq61}
\frac{S(a,b) + 9L(a,b)}{10} \le \frac{2N_3 (a,b) + 3L(a,b)}{5} \le N_1
(a,b)
\end{equation}

\noindent respectively.

\bigskip
Combining the above two inequalities (\ref{eq60}) and (\ref{eq61}), we can get (\ref{eq59})
provided the following hold
\begin{equation}
\label{eq62}
\frac{2N_3 (a,b) + 3L(a,b)}{5} \le \frac{5A(a,b) + 7L(a,b)}{12}
\end{equation}

\bigskip
In order to prove inequalities (\ref{eq62}), let us consider the difference
\begin{align}
 T_2 (a,b) & = \frac{5A(a,b) + 7L(a,b)}{12} - \frac{2N_3 (a,b) + 3L(a,b)}{5} \notag\\
& = \frac{25A(a,b) - 24N_3 (a,b) - L(a,b)}{60} = a\;f_{T_2 } \left( {\frac{b}{a}} \right), \notag
\end{align}
\noindent where
\[ f_{T_2 } (x)  = \frac{5}{12}\left( {\frac{x + 1}{2}} \right) -
\frac{1}{60}\left( {\frac{x - 1}{\ln x}} \right) - \frac{2}{5}\left(
{\frac{\left( {1 + \sqrt x } \right)\left( {\sqrt {2x + 2} } \right)}{4}}
\right), \quad x \ne 1, \quad x > 0 \]

\bigskip
The graph of the function $f_{T_2 } (x),\;\;x > 0,\;x \ne 1$ is given by

\begin{center}
\includegraphics[bb=0mm 0mm 208mm 296mm, width=45.6mm, height=45.6mm, viewport=3mm 4mm 205mm 292mm]{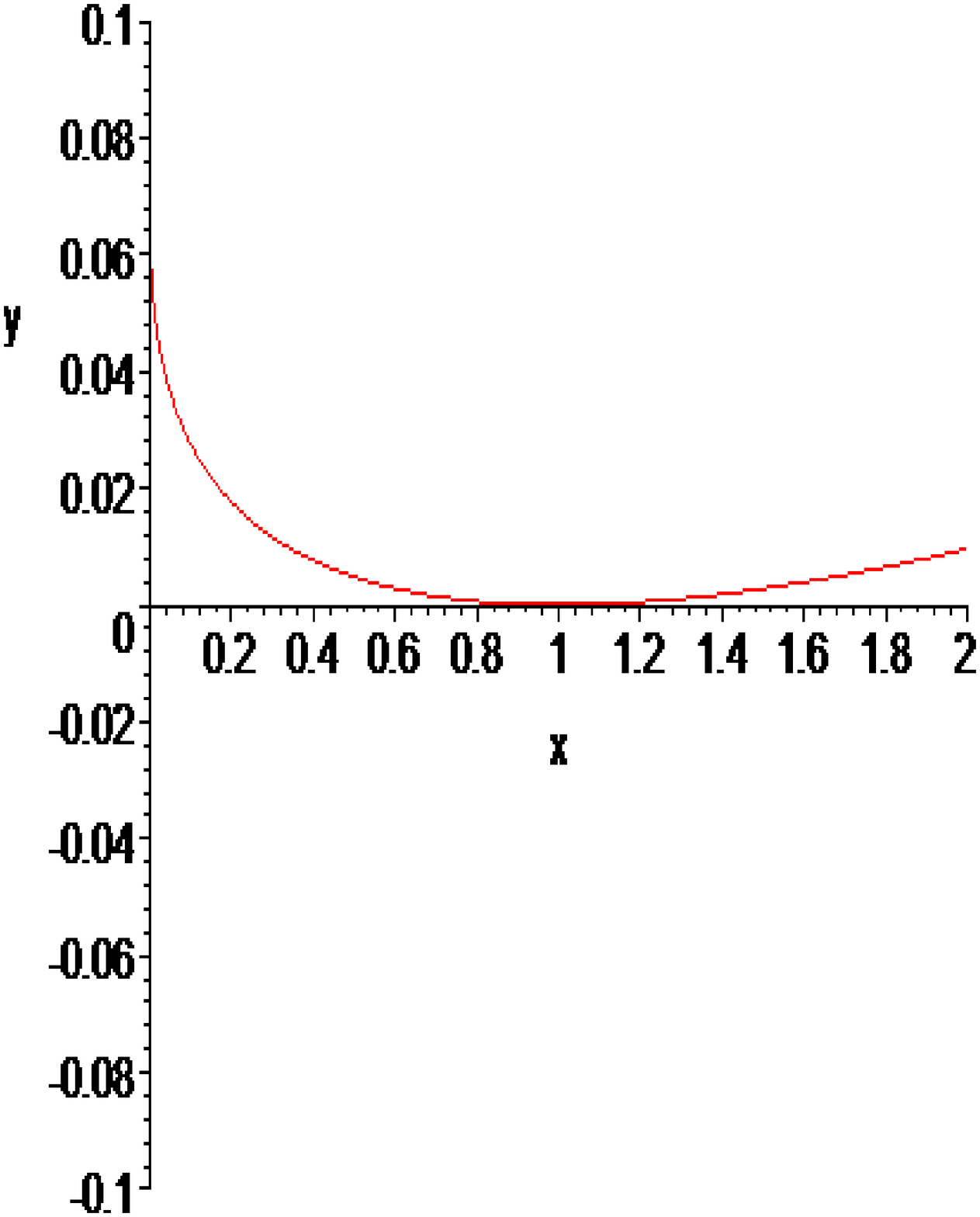}
\end{center}

From the above graph we observe that the function $f_{T_2 } (x),\;\;x >
0,\;x \ne 1$, is nonnegative with $f_{T_2 } (1) = 0$. Thus we conclude that
the validity of the expression (\ref{eq54}) proving the inequalities (\ref{eq51}).
\end{proof}

\begin{remark} There is no relation between the expressions
\[ \frac{S(a,b) + 5L(a,b)}{6} \]
\noindent and
\[\frac{2N_3 (a,b) + 3L(a,b)}{5}.\]
\end{remark}

\begin{proof} This we shall show by simple example. Let us write the difference
\[ T_3 (a,b) = \frac{S(a,b) + 5L(a,b)}{6} - \frac{2N_3 (a,b) + 3L(a,b)}{5}
 = a\;f_{T_3 } \left( {\frac{b}{a}} \right)\]

\noindent where
\[f_{T_3 } (x)  = \frac{\sqrt {2 + 2x^2} }{12} + \frac{7}{30}\left( {\frac{x
- 1}{\ln x}} \right) - \frac{\left( {1 + \sqrt x } \right)\sqrt {2x + 2}
}{10}, \quad x \ne 1, \quad x > 0\]

\bigskip
By simple calculations, we have

\begin{center}
$f_{T_3 } (0.00001) = - 0.00337512758$ and $f_{T_3 } (\ref{eq1}) = 0.0001321351$.
\end{center}

Thus we have two values, one negative and another positive proving that we
can't have relation between these two measures.
\end{proof}

\begin{remark} The followings hold
\begin{equation}
\label{eq63}
N_2 (a,b) \le \frac{5N_2 (a,b) + L(a,b)}{6} \le N_3 (a,b).
\end{equation}
\end{remark}

\begin{proof} In view of (\ref{eq60}), it is sufficient to show that
\[
N_3 (a,b) - \frac{5N_2 (a,b) + L(a,b)}{6} \ge 0
\]

Let us write
\[
T_4 (a,b) = N_3 (a,b) - \frac{5N_2 (a,b) + L(a,b)}{6} = a\;f_{T_3 } \left(
{\frac{b}{a}} \right),\]
\noindent where
\[f_{T_4 } (x)  = \frac{\left( {1 + \sqrt x } \right)\left( {\sqrt {2x + 2} }
\right)}{4} - \frac{5\left( {1 + \sqrt x + x} \right)}{18} -
\frac{1}{6}\frac{x - 1}{\ln x}, \quad x \ne 1, \quad x > 0 \]

The graph of the function $f_{T_4 } (x)$ is given by

\begin{center}
\includegraphics[bb=0mm 0mm 208mm 296mm, width=45.6mm, height=45.6mm, viewport=3mm 4mm 205mm 292mm]{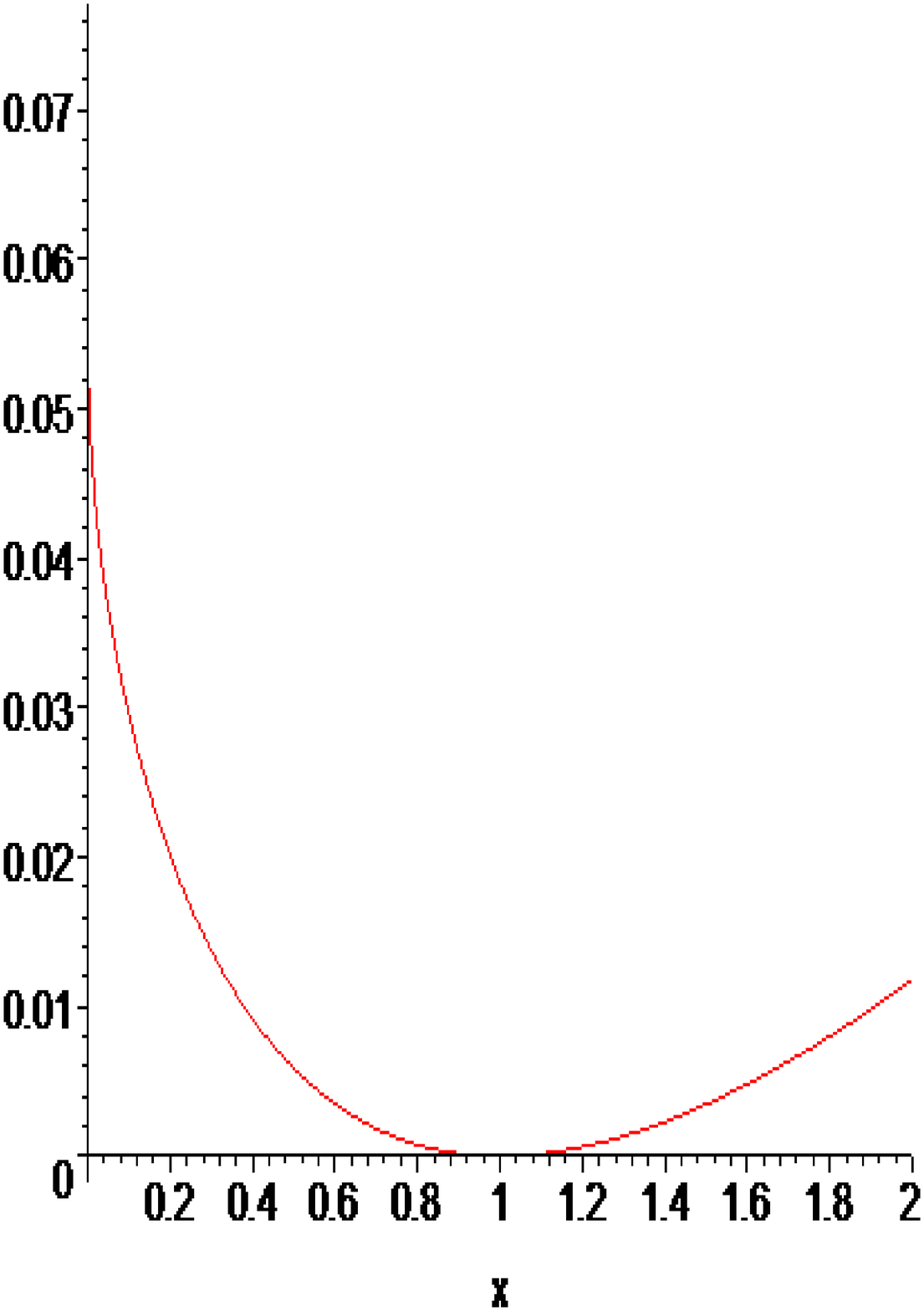}
\end{center}

From the graph above it is clear that $T_4 (a,b) \ge 0$. Proving the
inequalities (\ref{eq63}).
\end{proof}

\end{document}